%% file: paper.tex
\documentclass[runningheads]{llncs}
\usepackage{graphicx}
\usepackage{hyperref}
\usepackage{relsize}
\usepackage{xspace}
\usepackage{amsmath}
\usepackage{amsfonts}
\usepackage{amssymb}
\usepackage{multirow}
\usepackage{mathpartir}
\usepackage{xcolor}
\usepackage{latexsym}
\usepackage{tikz}
\usepackage{cancel}
\usetikzlibrary{arrows,decorations.pathmorphing}
\usetikzlibrary{positioning,calc}

\usepackage{booktabs}
\usepackage{tabu}
\usepackage{multirow}
\usepackage{url}
\usepackage{dcolumn}

\newcommand\definetool[2]{\newcommand{#1}{\texorpdfstring{{\scshape #2}}{#2}\xspace}}
\definetool{\ultimate}     {Ultimate}
\definetool{\ltlautomizer} {Ultimate LTLAutomizer}
\definetool{\automizer}    {Ultimate Automizer}
\definetool{\taipan}       {Ultimate Taipan}
\definetool{\smtinterpol}  {SMTInterpol}
\definetool{\cpachecker}   {CPAchecker}
\definetool{\blast}        {Blast}
\definetool{\slam}         {SLAM}
\definetool{\staipan}      {Taipan}
\definetool{\rstaipan}     {Automizer}
\definetool{\lstaipan}     {LazyTaipan}
\definetool{\zzz}          {Z3}
\definetool{\cvc}          {CVC4}

\newcommand{\svcomp}		{{SV-COMP}\xspace}

\newif\iflongversion
\longversionfalse
\longversiontrue
\iflongversion
\usepackage{johofooter}
\fi

\newcommand\mT{\mathcal{T}}
\newcommand\T{$\mT$\xspace}
\newcommand\Arr{\ensuremath{\mathcal{A}}\xspace}
\newcommand\tarr{\ensuremath{\mathcal{T}_\Arr}\xspace}

\newcommand\store[3]{\ensuremath{#1\langle #2\lhd#3\rangle}\xspace}
\newcommand\select[2]{\ensuremath{#1[#2]}\xspace}
\newcommand\strongeq{\ensuremath{\sim}\xspace}
\newcommand\strongedge{\ensuremath{\leftrightarrow}}
\newcommand\weakedgei[1][i]{\ensuremath{\overset{\smash{#1}}{\strongedge}}}
\newcommand\selectedgei[1][i]{\mathbin{\tikz[baseline] \draw[<->,dashed] (0pt,.65ex) -- node[above,inner sep=0cm] {\scriptsize$#1$} (2.2em,.65ex);}}
\newcommand\weakpath[1][P]{\ensuremath{\overset{\smash{#1}}{\Leftrightarrow}}}
\newcommand\weakeqi[1][i]{\ensuremath{\approx_{#1}}\xspace}
\newcommand\weakcongi[1][i]{\ensuremath{\sim_{#1}}}
\newcommand\Stores[1]{\ensuremath{\mathop{\mathrm{Stores}}\left(#1\right)}}
\newcommand\Cond[1]{\ensuremath{\mathop{\mathrm{Cond}}(#1)}}

\newcommand\axdiff{\ensuremath{\mathit{AxDiff}}}
\newcommand\taxdiff{\ensuremath{\mT_\axdiff}\xspace}
\newcommand\diff[2]{\ensuremath{\mathop{\mathrm{diff}}(#1,#2)}\xspace}
\newcommand\tdiff{$\mathrm{diff}$\xspace}
\newcommand\rewrite[1][1]{\mathbin{\overset{\smash{#1}}{\rightsquigarrow}}}
\newcommand\smashrewrite[1][1]{\mathbin{\smash{\overset{\smash{#1}}{\rightsquigarrow}}}}

\newcommand\project[1]{\ensuremath{\!\downharpoonright\!#1}\xspace}
\newcommand\EQ[2]{\ensuremath{\mathop{\mathrm{EQ}}(#1,#2)}}
\newcommand\tEQ{$\mathrm{EQ}$\xspace}

\newcommand\weq[4]{\ensuremath{\mathop{\mathrm{weq}}(#1,#2,#3,#4)}}
\newcommand\bigweq[4]{\ensuremath{\mathop{\mathrm{weq}}\Big(#1,#2,#3,#4\Big)}}
\newcommand\nweq[4]{\ensuremath{\mathop{\mathrm{nweq}}(#1,#2,#3,#4)}}
\newcommand\bignweq[4]{\ensuremath{\mathop{\mathrm{nweq}}\Big(#1,#2,#3,#4\Big)}}

\newcommand\mI{\ensuremath{I}}
\newcommand\I{$\mI$\xspace}
\newcommand\A{$A$\xspace}
\newcommand\B{$B$\xspace}
\newcommand\smtlet{\mathop{\mathsf{let}}}
\def\ourtitle{Efficient Interpolation for the Theory of Arrays}

\title{\ourtitle}

\hypersetup{pdfauthor={Jochen Hoenicke, Tanja Schindler},
  pdftitle={\ourtitle}}
\author{Jochen Hoenicke \and Tanja Schindler%
	\thanks{This work is supported by the
	German Research Council (DFG) under HO 5606/1-1.}}

\institute{
  University of Freiburg\\
  \email{\{hoenicke,schindle\}@informatik.uni-freiburg.de}}
  
\authorrunning{Hoenicke and Schindler}
\titlerunning{Efficient Interpolation for the Theory of Arrays}

\begin{document}
\maketitle

\iflongversion
\johofooter{lncs}{\emph{IJCAR 2018}}{10.1007/978-3-319-94205-6\_36}
\fi
\input{intro}
\input{section-notation}

\input{section-preliminaries}

\input{section-readoverweakeq}

\input{subsection-interpolants-readoverweakeq-shared}
\input{subsection-interpolants-readoverweakeq-shared-a}
\input{subsection-interpolants-readoverweakeq-blocal}
\input{subsection-interpolants-readoverweakeq-alocal}

\input{section-weakeqext}

\input{subsection-interpolants-weakeqext-blocal}
\input{subsection-interpolants-weakeqext-alocal}
\input{subsection-interpolants-weakeqext-mixed}

\input{complexity}

\input{evaluation}
\input{conclusion}

\bibliographystyle{plain}
\bibliography{paper-short}

\end{document}

%% file: intro.tex
\begin{abstract}
  Existing techniques for Craig interpolation for the quantifier-free fragment of the theory of arrays are inefficient for computing sequence and tree interpolants:
  the solver needs to run for every partitioning $(A, B)$ of the interpolation problem to avoid creating $AB$-mixed terms.
  We present a new approach using Proof Tree Preserving Interpolation and an array solver based on Weak Equivalence on Arrays.
  We give an interpolation algorithm for the lemmas produced by the array solver.
  The computed interpolants have worst-case exponential size for extensionality lemmas and worst-case quadratic size otherwise.
  We show that these bounds are strict in the sense that there are lemmas with no smaller interpolants.
  We implemented the algorithm and show that the produced interpolants are useful to prove memory safety for C programs.
\end{abstract}

\section{Introduction}
\label{sec:introduction}

Several model-checkers~\cite{DBLP:conf/tacas/AndrianovFMMV17,DBLP:journals/sttt/BeyerHJM07,DBLP:conf/tacas/CassezSRPSM17,DBLP:conf/tacas/DanglLW15,DBLP:conf/tacas/GreitschusDHNSS17,DBLP:conf/tacas/HeizmannCDGNMSS17,DBLP:conf/popl/HenzingerJMS02,DBLP:conf/cav/McMillan06,DBLP:conf/tacas/NutzDMP15}
use interpolants to find candidate invariants to prove the correctness of software.
They require efficient tools to check satisfiability of a formula in a decidable theory and to compute interpolants (usually sequence or tree interpolants) for unsatisfiable formulas.
Moreover, they often need to combine several theories, e.g., integer or bitvector theory for reasoning about numeric variables and array theory for reasoning about pointers.
In this paper we present an interpolation procedure for the quantifier-free fragment of the theory of arrays that allows for the combination with other theories and that reuses an existing unsatisfiability proof to compute interpolants efficiently.

Our method is based on the array solver presented in~\cite{DBLP:conf/frocos/ChristH15}, which fits well into existing Nelson-Oppen frameworks.
The solver generates lemmas, valid in the theory of arrays, that explain equalities between terms shared between different theories.
The terms do not necessarily belong to the same formula in the interpolation problem and the solver does not need to know the partitioning.
Instead, we use the technique of Proof Tree Preserving Interpolation~\cite{DBLP:conf/tacas/ChristHN13}, which produces interpolants from existing proofs that can contain propagated equalities between symbols from different parts of the interpolation problem.

The contribution of this paper is an algorithm to interpolate the lemmas produced by the solver of the theory of arrays without introducing quantifiers.
The solver only generates two types of lemmas, namely a variant of the read-over-write axiom and a variant of the extensionality axiom.
However, the lemmas contain array store chains of arbitrary length which need to be handled by the interpolation procedure.
The interpolants our algorithm produces summarize array store chains, e.\,g., they state that two shared arrays at the end of a sub-chain differ at most at $m$ indices, each satisfying a subformula.
Bruttomesso et al.~\cite{DBLP:journals/corr/abs-1204-2386} showed that adding a diff function to the theory of arrays makes the quantifier-free fragment closed under interpolation, i.e. it ensures the existence of quantifier-free interpolants for quantifier-free problems.
We use the diff function to obtain the indices for store chains and give a more efficient algorithm that exploits the special shape of the lemmas provided by the solver.

Nevertheless, the lemma interpolants produced by our algorithm may be exponential in size (with respect to the size of the input lemma).
We show that this is unavoidable as there are lemmas that have no small interpolants.  

\paragraph{Related Work.}
\label{sec:relatedwork}
The idea of computing interpolants from resolution proofs goes back to~Kraj\'i\v{c}ek and Pudl\'ak~\cite{DBLP:journals/jsyml/Krajicek97,DBLP:journals/jsyml/Pudlak97}.
McMillan~\cite{DBLP:journals/tcs/McMillan05} extended their work to SMT with a single theory.
The theory of arrays can be added by including quantified axioms and can be interpolated using, e.g., the method by Christ and Hoenicke~\cite{DBLP:conf/dagstuhl/ChristH10} for quantifier instantiation,
or the method of Bonacina and Johansson~\cite{DBLP:journals/jar/BonacinaJ15} for superposition calculus.
Brillout~et~al~\cite{DBLP:conf/cade/BrilloutKRW10a} apply a similar algorithm to compute interpolants from sequent calculus proofs.
In contrast to our approach, using such a procedure generates quantified interpolants.

Equality interpolating theories~\cite{DBLP:conf/cade/YorshM05,DBLP:journals/tocl/BruttomessoGR14} allow for the generation of quantifier-free interpolants in the combination of quantifier-free theories.
A theory is equality interpolating if it can express an interpolating term for each equality using only the symbols occurring in both parts of the interpolation problem.
The algorithm of Yorsh and Musuvathi~\cite{DBLP:conf/cade/YorshM05} only supports convex theories and is not applicable to the theory of arrays.
Bruttomesso~et~al.~\cite{DBLP:journals/tocl/BruttomessoGR14} extended the framework to non-convex theories.
They also present a complete interpolation procedure for the quantifier-free theory of arrays that works for theory combination in~\cite{DBLP:journals/corr/abs-1204-2386}.
However, their solver depends on the partitioning of the interpolation problem.
This can lead to exponential blow-up of the solving procedure.
Our interpolation procedure works on a proof produced by a more efficient array solver that is independent of the partitioning of the interpolation problem.

Totla and Wies~\cite{DBLP:journals/jar/TotlaW16} present an interpolation method for arrays based on complete instantiations.
It combines the idea of~\cite{DBLP:journals/tocl/BruttomessoGR14} with local theory extension~\cite{DBLP:conf/cade/Sofronie-Stokkermans05}.
Given an interpolation problem $A$ and $B$, they define two sets,
each using only symbols from $A$ resp.\ $B$, that contain the instantiations of the array axioms needed to prove unsatisfiability.
Then an existing solver and interpolation procedure for uninterpreted functions can be used to compute the interpolant.
The procedure causes a quadratic blow-up on the input formulas.
We also found that their procedure fails for some extensionality lemmas, when we used it to create candidate interpolants.
We give an example for this in Sect.~\ref{sec:complexity}.

The last two techniques require to know the partitioning at solving time.
Thus, when computing sequence~\cite{DBLP:journals/tcs/McMillan05} or tree interpolants~\cite{DBLP:conf/popl/HeizmannHP10}, they would require either an adapted interpolation procedure or the solver has to run multiple times.
In contrast, our method can easily be extended to tree interpolation~\cite{DBLP:journals/jar/ChristH16}.

%% file: section-notation.tex
\section{Basic Definitions}
\label{sec:notation}
We assume standard first-order logic.
A theory \T is given by a signature \(\Sigma\) and a set of axioms.
The theory of arrays \tarr is parameterized by an index theory and an element theory.
Its signature \(\Sigma_\Arr\) contains the \emph{select} (or \emph{read}) function \(\select{\cdot}{\cdot}\) and the \emph{store} (or \emph{write}) function \(\store{\cdot}{\cdot}{\cdot}\).
In the following, \(a, b, s, t\) denote array terms, \(i, j, k\) index terms and \(v, w\) element terms.
For array \(a\), index \(i\) and element \(v\), \(\select{a}{i}\) returns the element stored in \(a\) at \(i\), and \(\store{a}{i}{v}\) returns a copy of \(a\) where the element at index \(i\) is replaced by the element \(v\), leaving \(a\) unchanged.
The functions are defined by the following axioms proposed by McCarthy \cite{DBLP:conf/ifip/McCarthy62}.
\begin{align}
\forall a\ i\ v.\ \select{\store{a}{i}{v}}{i}&=v
\label{eq:selectoverstore_1}\tag{idx}\\
\forall a\ i\ j\ v.\ i\neq j \rightarrow\select{\store{a}{i}{v}}{j}&=\select{a}{j}
\label{eq:selectoverstore_2}\tag{read-over-write}
\end{align}
We consider the variant of the extensional theory of arrays proposed by Bruttomesso~et~al.~\cite{DBLP:journals/corr/abs-1204-2386} where the signature is extended by the function \(\diff{\cdot}{\cdot}\).
For distinct arrays \(a\) and \(b\), it returns an index where \(a\) and \(b\) differ, and an arbitrary index otherwise.
The extensionality axiom then becomes
\begin{equation}\label{eq:diff}
  \forall a\ b.\ \select{a}{\diff{a}{b}}=\select{b}{\diff{a}{b}}
  \rightarrow a=b \enspace.
\tag{ext-diff}
\end{equation}
The authors of \cite{DBLP:journals/corr/abs-1204-2386} have shown that the quantifier-free fragment of the theory of arrays with \tdiff, \taxdiff, is closed under interpolation.
To express the interpolants
conveniently, we use the notation from~\cite{DBLP:journals/jar/TotlaW16} for rewriting arrays.
For $m \geq 0$, we define $a\rewrite[\smash{m}] b$ for two arrays $a$ and $b$ inductively as
\begin{equation*}
  a\rewrite[0] b := a \qquad a\rewrite[m+1] b := \store{a}{\diff{a}{b}}{\select{b}{\diff{a}{b}}}\rewrite[m] b\enspace.
\end{equation*}
Thus, $a\rewrite[\smash{m}] b$ changes the values in $a$ at $m$ indices to the values stored in $b$.
The equality $a \rewrite[\smash{m}] b = b$ holds if and only if $a$ and $b$ differ at up to $m$ indices.
The indices where they differ are the diff terms occurring in $a\rewrite[\smash{m}]b$.

An interpolation problem $(A,B)$ is a pair of formulas where $A\land B$ is unsatisfiable.
A \emph{Craig interpolant} for $(A,B)$ is a formula \I such that
(i) \A implies \I in the theory \T,
(ii) \I and \B are \T-unsatisfiable and
(iii) all non-theory symbols occurring in \I are shared between \A and \B.
Given an interpolation problem $(A,B)$, the symbols shared between \A and \B are called \emph{shared}, symbols only occurring in \A are called \emph{\A-local} and symbols only occurring in \B, \emph{\B-local}.
A literal, e.g.  $a=b$, that contains \A-local and \B-local symbols is called \emph{mixed}.

%% file: section-preliminaries.tex
\section{Preliminaries}
\label{sec:preliminaries}

Our interpolation procedure operates on theory lemmas instantiated from particular variants of the read-over-write and extensionality axioms, and is designed to be used within the proof tree preserving interpolation framework.
In the following, we give a short overview of this method and revisit the definitions and results about weakly equivalent arrays.

\subsection{Proof Tree Preserving Interpolation}
\label{sec:prooftreepreservinginterpolation}

The proof tree preserving interpolation scheme presented by Christ~et~al.~\cite{DBLP:conf/tacas/ChristHN13} allows to compute interpolants for an unsatisfiable formula using a resolution proof that is unaware of the interpolation problem.

For a partitioning $(A,B)$ of the interpolation problem, two projections $\cdot\project{A}$ and $\cdot\project{B}$ project a literal to its \A-part resp.\ \B-part.
For a literal $\ell$ occurring in~\A, we define \(\ell\project{A}\equiv \ell\).
If $\ell$ is \A-local, \(\ell\project{B}\equiv \mathbf{true}\).
For $\ell$ in \B, the projections are defined analogously.
These projections are canonically extended to conjunctions of literals.
A \emph{partial interpolant} of a clause $C$ occurring in the proof tree is defined as the interpolant of $A \land (\lnot C) \project{A}$ and $B \land (\lnot C)\project{B}$.
Partial interpolants can be computed inductively over the proof tree and the partial interpolant of the root is the interpolant of $A$ and $B$.
For a theory lemma $C$, a partial interpolant is computed for the interpolation problem
\(((\lnot C)\project{A},(\lnot C)\project{B})\).

The core idea of proof tree preserving interpolation is a scheme to handle mixed equalities.
For each \(a=b\) where \(a\) is \A-local and $b$ is \B-local, a fresh variable $x_{ab}$ is introduced.
This allows to define the projections as follows.
\[ (a = b) \project{A}  \equiv (a=x_{ab}) \qquad (a=b) \project{B} \equiv (x_{ab} = b) \]
Thus, $a=b$ is equivalent to \(\exists x_{ab}. (a = b) \project{A} \land (a = b) \project{B}\) and $x_{ab}$ is a new shared variable that may occur in partial interpolants.
For disequalities we introduce an uninterpreted predicate EQ 
and define the projections for $a\neq b$ as
\[ (a \neq b) \project{A} \equiv \EQ{x_{ab}}{a} \qquad (a \neq b) \project{B} \equiv \lnot \EQ{x_{ab}}{b}\enspace. \]
For an interpolation problem $(A\land (\lnot C) \project{A}, B\land (\lnot C) \project{B})$ where $\lnot C$ contains $a\neq b$, we require as additional symbol condition that $x_{ab}$ only occurs as first parameter of an \tEQ predicate which occurs positively in the interpolant, i.e., the interpolant has the form \(I[\EQ{x_{ab}}{s_1}]\dots[\EQ{x_{ab}}{s_n}]\)\footnote{One can show that such an interpolant exists for every equality interpolating theory in the sense of Definition~4.1 in \cite{DBLP:journals/tocl/BruttomessoGR14}.  The terms $s_i$ are the terms \underline{v} in that definition.}.
For a resolution step on the mixed pivot literal $a=b$, the following rule combines the partial interpolants of the input clauses to a partial interpolant of the resolvent.
\[
\inferrule {C_1\lor a=b:I_1[\EQ{x_{ab}}{s_1}]\dots[\EQ{x_{ab}}{s_n}] \\ C_2 \lor a\neq b: I_2(x_{ab}) } 
           {C_1 \lor C_2: I_1[I_2(s_1)]\dots[I_2(s_n)] }
\]

\subsection{Weakly Equivalent Arrays}
\label{sec:weaklyequivalentarrays}
Proof tree preserving interpolation can handle mixed literals, but it cannot deal with mixed \emph{terms} which can be produced when instantiating \eqref{eq:selectoverstore_2} on an \A-local store term and a \B-local index.
The lemmas produced in the decision procedure for the theory of arrays presented by Christ and Hoenicke~\cite{DBLP:conf/frocos/ChristH15} avoid such mixed terms by exploiting \emph{weak equivalences} between arrays.

For a formula \(F\), let \(V\) be the set of terms that contains the array terms in $F$ and in addition the select terms $\select{a}{i}$ and their indices $i$ and for each store term $\store{a}{i}{v}$ in $F$ the terms $i$, $v$, \(\select{a}{i}\) and $\select{\store{a}{i}{v}}{i}$.
Let \(\strongeq\) be the equivalence relation on \(V\) representing equality.
The \emph{weak equivalence graph} \(G^W\) is defined by its vertices, the array-valued terms in \(V\), and its undirected edges of the form (i) \(s_1\strongedge s_2\) if \(s_1\strongeq s_2\) and (ii) \(s_1\weakedgei[\smash{i}] s_2\) if \(s_1\) has the form \(\store{s_2}{i}{\cdot}\) or vice versa.
If two arrays \(a\) and \(b\) are connected in \(G^W\) by a path \(P\), they are called \emph{weakly equivalent}.
This is denoted by \(a\weakpath b\).
Weakly equivalent arrays can differ only at finitely many positions given by \(\Stores{P} := \{i\ |\ \exists s_1\ s_2.\ s_1 \weakedgei[\smash{i}] s_2 \in P\}\).
Two arrays \(a\) and \(b\) are called \emph{weakly equivalent on \(i\)}, denoted by \(a\weakeqi b\), if they are connected by a path \(P\) such that $k\not\strongeq i$ holds for each $k\in\Stores{P}$.
Two arrays \(a\) and \(b\) are called \emph{weakly congruent on \(i\)}, \(a\weakcongi b\), if they are weakly equivalent on \(i\), or if there exist \(\select{a'}{j}, \select{b'}{k} \in V\) with \(\select{a'}{j}\strongeq\select{b'}{k}\) and $j\strongeq k \strongeq i$ and $a' \weakeqi a$, $b' \weakeqi b$. 
If \(a\) and \(b\) are weakly congruent on \(i\), they must store the same value at~\(i\).
For example, if \(\store{a}{i+1}{v} \strongeq b\) and \(\select{b}{i}\strongeq\select{c}{i} \), arrays \(a\) and \(b\) are weakly equivalent on \(i\) while \(a\) and \(c\) are only weakly congruent on \(i\).

We use \(\Cond{a\weakpath b},\Cond{a \weakeqi b},\Cond{a \weakcongi b}\) to denote the conjunction of the literals $v = v'$ (resp.\ $v \neq v'$), $v,v'\in V$, such that $v\strongeq v'$ (resp. $v\not\strongeq v'$) is necessary to show the corresponding property.
Instances of array lemmas are generated according to the following rules:
\begin{equation}
\inferrule{a\weakeqi b \\ i\strongeq j \\ \select{a}{i},\select{b}{j}\in V}
		{\Cond{a\weakeqi b} \land i = j \rightarrow \select{a}{i}=\select{b}{j}}
\tag{roweq}\label{eq:rule_readoverweakeq}
\end{equation}
\begin{equation}
\inferrule{a\weakpath b \\ \forall i\in\Stores{P}.\ a\weakcongi b \\ a,b \in V}
		{\Cond{a\weakpath b} \land \bigwedge_{i\in\Stores{P}}\hspace{-1.5em} \Cond{a\weakcongi b} \rightarrow a=b}
\tag{weq-ext}\label{eq:rule_weakeqext}
\end{equation}
The first rule, based on \eqref{eq:selectoverstore_2}, propagates equalities between select terms and the second, based on extensionality, propagates equalities on array terms.
These rules are complete for the quantifier-free theory of arrays~\cite{DBLP:conf/frocos/ChristH15}.
In the following, we describe how to derive partial interpolants for these lemmas.

%% file: section-readoverweakeq.tex
\section{Interpolants for Read-Over-Weakeq Lemmas}
\label{sec:interpolants_readoverweakeq}

A lemma generated by \eqref{eq:rule_readoverweakeq} explains the conflict (negation of the lemma)
\begin{equation*}
  \Cond{a \weakeqi b} \land i = j \land \select{a}{i} \neq \select{b}{j}\enspace .
\end{equation*}%
The weak equivalence $a \weakeqi b$ ensures that \(a\) and \(b\) are equal at $i=j$ which contradicts \(\select{a}{i}\neq\select{b}{j}\) (see Fig.~\ref{fig:readoverweakeq}).

\begin{figure}
\centering
\begin{tikzpicture}
	\node (ai) at (0,2) {\(a[i]\)};
	\node (bj) at (8,2) {\(b[j]\)};
	\node (a) at (0,0) {\(a\)};
	\node (dots) at (4,0) {\(\cdots\)};
	\node (b) at (8,0) {\(b\)};
	\node (i) at (2.5,1.3) {\(i\)};
	\node (j) at (5.5,1.3) {\(j\)};
	\draw [-] (ai) -- node {$|$} (bj);
	\draw [dotted, thick] (ai) -- (a) (ai) -- (i);
	\draw [dotted, thick] (bj) -- (b) (bj) -- (j);
	\draw [-] (i) -- (j);
\draw [decorate, decoration=zigzag] (a) -- (dots) node [above,near start] (k1) {\(k_1\)} node [above,near end] (k2) {\(k_2\)} -- (b) node (kmminus1)[above,near start] {\(k_{m-1}\)} node (km) [above,near end] {\(k_m\)};
	\draw [-] (i) -- node {\(\backslash\)} (k1);
	\draw [-] (i) -- node[pos=.6] {---} (k2);	
	\draw [-] (i) -- node[pos=.7] {\(/\)} (kmminus1);
        \draw [-] (i) -- node {\(/\)} (km);
\end{tikzpicture}
\caption{A read-over-weakeq conflict.
Solid lines represent strong \mbox{(dis-)equalities}, dotted lines function-argument relations, and zigzag lines represent weak paths consisting of store steps and array equalities.}
\label{fig:readoverweakeq}
\end{figure}
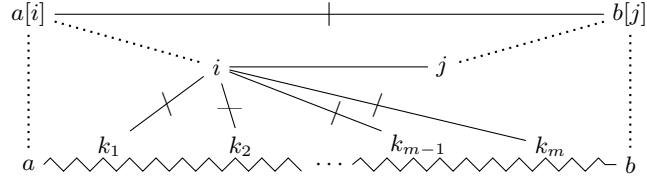%
The general idea for computing an interpolant for this conflict, similar to~\cite{DBLP:journals/corr/abs-1111-5652}, is to summarize maximal paths induced by literals of the same part (\A or \B), relying on the fact that the terms at the ends of these paths are shared.
If a shared term is equal to the index $i$, we can express that the shared arrays at the path ends coincide or must differ at the index.
There is a \emph{shared term for $i=j$} if $i$ or $j$ are shared or if $i=j$ is mixed. 
If there is no shared term for $i=j$, the interpolant can be expressed using diff chains to capture the index.
We identify four basic cases:
(i) there is a shared term for \(i=j\) and \(\select{a}{i}=\select{b}{j}\) is in \B or mixed,
(ii) there is a shared term for \(i=j\) and \(\select{a}{i}=\select{b}{j}\) is \A-local,
(iii) both \(i\) and \(j\) are \B-local, and
(iv) both \(i\) and \(j\) are \A-local.

%% file: subsection-interpolants-readoverweakeq-shared.tex
\subsection{Shared Term for $i=j$ and $\select{a}{i}=\select{b}{j}$ is in \B or mixed} 
\label{sec:interpolants_readoverweakeq_shared}

If there exists a shared term \(x\) for the index equality \(i=j\), the interpolant can contain terms \(\select{s}{x}\) for shared array terms \(s\) occurring on the weak path between $a$ and $b$.
The basic idea is to summarize the weak \A-paths by applying rule \eqref{eq:rule_readoverweakeq} on their end terms.

\begin{example}\label{the:motivatingexample_readoverweakeq_shared}Consider the following read-over-weakeq conflict:
  \begin{align*}
	&\phantom{{}\land{}} a=s_1 \land \store{s_1}{k_1}{v_1}=s_2 \land \store{s_2}{k_2}{v_2}=s_3 \land s_3=b\\
	&\land i\neq k_1 \land i\neq k_2 \land i=j \land \select{a}{i}\neq\select{b}{j}
  \end{align*}
where \(a\), \(k_2\), \(v_2\), \(i\) are \A-local, \(b\), \(k_1\), \(v_1\), \(j\) are \B-local, and \(s_1\), \(s_2\), \(s_3\) are shared.
Projecting the mixed literals on \A and \B as described in Sect.~\ref{sec:prooftreepreservinginterpolation} yields the interpolation problem
\begin{align*}
A &:  a=s_1 \land \store{s_2}{k_2}{v_2}=s_3 \land \EQ{x_{i{k_1}}}{i} \land i\neq k_2 \land i=x_{ij} \land \EQ{x_{\select{a}{i}\select{b}{j}}}{\select{a}{i}}\\
B &: \store{s_1}{k_1}{v_1}=s_2 \land s_3=b \land \lnot\EQ{x_{i{k_1}}}{k_1} \land x_{ij}=j \land \lnot\EQ{x_{\select{a}{i}\select{b}{j}}}{\select{b}{j}}\enspace.
\end{align*}
An interpolant is
\(\mI\equiv \EQ{x_{\select{a}{i}\select{b}{j}}}{\select{s_1}{x_{ij}}} \land \select{s_2}{x_{ij}}=\select{s_3}{x_{ij}} \land \EQ{x_{i{k_1}}}{x_{ij}}\).

\end{example}

\paragraph{Algorithm.}
The first step is to subdivide the weak path \(P:a\weakeqi b\) into \A- and \B-paths.
An equality edge $\strongedge$ is assigned to either an \A- or \B-path depending on whether the corresponding equality is in \A or \B.
A mixed equality $a'=b'$ is split into the \A-local equality $a'=x_{a'b'}$ and the \B-local equality $x_{a'b'} = b'$.
Store edges $\weakedgei$ are assigned depending on which part contains the store term.
If an equality or store term is shared between both parts, the algorithm can assign it to \A or \B arbitrarily.
The whole path from $a$ to $b$ is then an alternation of \A- and \B-paths, which meet at shared boundary terms.

Let \(x\) be the shared term for \(i=j\), i.e.\ \(x\) stands for \(i\) if \(i\) is shared, for \(j\) if \(i\) is not shared but \(j\) is, and for the auxiliary variable \(x_{ij}\) if \(i=j\) is mixed.

\noindent
(i) An inner \A-path $\pi: s_1 \weakeqi s_2$ of $P$ starts and ends with a shared term.
The summary is $s_1[x] = s_2[x]$.
For a store edge on $\pi$ with index $k$, add the disjunct $x = k$ if the corresponding disequality $i\neq k$ is \B-local, and the disjunct $\EQ{x_{ik}}{x}$ if the disequality is mixed.
The interpolant of the subpath is
\begin{equation*}
  I_\pi \equiv s_1[x] = s_2[x] \lor F^A_\pi(x) \qquad \text{where }
F^A_\pi(x) \equiv \hspace{-1em}\bigvee_{\substack{k\in\Stores{\pi}\\ i\neq k\text{ \B-local}}} \hspace{-1em}x=k \hspace{.5em}\lor \hspace{-1em}\bigvee_{\substack{k\in\Stores{\pi}\\ i\neq k\text{ mixed}}}\hspace{-1em}\EQ{x_{ik}}{x}\enspace.
\end{equation*}
(ii) If \(\select{a}{i}\neq\select{b}{j}\) is mixed and \(\select{a}{i}\) is \A-local, the first \A-path on \(P\) starts with \(a\) or \(a\) is shared, i.e.
$\pi: a \weakeqi s_1$ (where $s_1$ can be $a$).
For the path \(\pi\), build the term \(\EQ{x_{\select{a}{i}\select{b}{j}}}{\select{s_1}{x}}\) and add $F^A_{\pi}(x)$ as in case (i).
\begin{equation*}
  I_\pi \equiv \EQ{x_{\select{a}{i}\select{b}{j}}}{\select{s_1}{x}} \lor F^A_{\pi}(x)
\end{equation*}
(iii) Similarly in the case where \(\select{a}{i}\neq\select{b}{j}\) is mixed and \(\select{b}{j}\) is \A-local, the last \A-path on \(P\) ends with \(b\) or \(b\) is shared,  $\pi: s_n \weakeqi b$.  In this case the disjunct $i\neq j$ needs to be added if $i=j$ is \B-local and $i$, $j$ are both shared.
\begin{equation*}
  I_\pi \equiv \EQ{x_{\select{a}{i}\select{b}{j}}}{\select{s_n}{x}} \lor F^A_{\pi}(x)\;\;[{}\lor i \neq j]
\end{equation*}
(iv) For every \B-path \(\pi\), add the conjunct \(x\neq k\) for each \A-local index disequality \(i\neq k\), and the conjunct \(\EQ{x_{ik}}{x}\) for each mixed index disequality \(i\neq k\) on \(\pi\).
We define
\begin{equation*}
F^B_{\pi}(x) \equiv \hspace{-1em}\bigwedge_{\substack{k\in\Stores{\pi}\\ i\neq k\text{ \A-local}}}\hspace{-1em} x\neq k \hspace{.5em} \land \hspace{-1em}\bigwedge_{\substack{k\in\Stores{\pi}\\ i\neq k\text{ mixed}}}\hspace{-1em}\EQ{x_{ik}}{x}\enspace.
\end{equation*}

The lemma interpolant is the conjunction of the above path interpolants.
If \(i,j\) are shared, \(\select{b}{j}\) is in \B, and \(i=j\) is \A-local, add the conjunct \(i=j\).

\begin{lemma}\label{the:correctness_readoverweakeq_shared}
  If $x$ is a shared term for $i=j$ and $\select{a}{i} = \select{b}{j}$ is in $B$ or mixed, a partial interpolant of
  the lemma
  $\Cond{a \weakeqi b} \land i = j \rightarrow \select{a}{i} = \select{b}{j}$
  is
\begin{equation*}\label{eq:interpolants_readoverweakeq_shared}
\mI\equiv\bigwedge_{\pi\in\text{\A-paths}}\hspace{-1em}\mI_{\pi}\quad\land\quad\bigwedge_{\pi\in\text{\B-paths}} \hspace{-1em} F^B_\pi(x) \quad[{}\land i=j{}]\enspace.
\end{equation*}
\end{lemma}

\begin{proof}
  The interpolant only contains the shared boundary arrays, the shared term $x$ for $i=j$, auxiliary variables for mixed disequalities under an EQ predicate, and shared store indices $k$ where the store term is in a different part than the corresponding index disequality.
  
\(\lnot C\project{A}\) implies \(\mI\):
For a \B-path \(\pi\), we show that \(F^B_\pi(x)\) follows from the \A-part.
If \(i\) is \B-local, there are no \A-local or mixed index disequalities and \(F^B_\pi(x)\) holds trivially.
Otherwise \(i=x\) follows from \A, since either \(i\) is shared and \(x\) is \(i\), \(i=j\) is \A-local and \(x\) is \(j\), or \(i=x\) is the \A-projection of the mixed equality \(i=j\).
Then \(F^B_\pi(x)\) follows by replacing \(i\) by \(x\) in \A-local disequalities and \A-projections of mixed disequalities on \(\pi\).
For an \A-path \(\pi\), if \(F^A_\pi(x)\) does not hold, we get \(\select{s_1}{x}=\select{s_2}{x}\) by applying rule \eqref{eq:rule_readoverweakeq}.
Note that \(x\neq k\) follows from \(i=x\) if \(i \neq k\) is \A-local, and from \(\EQ{x_{ik}}{k}\) and \(\lnot F^A_\pi(x)\) in the mixed case.
For the outer \A-path in case (ii), \(\select{a}{x} = \select{s_1}{x}\) is combined with the \A-projection of the mixed disequality \(\select{a}{i}\neq\select{b}{j}\) using \(i=x\), which yields the \tEQ term.
Analogously we get the \tEQ term for (iii), but to derive \(j=x\) in the case where both \(i\) and \(j\) are shared but \(i=j\) is \B-local, we need to exclude \(i \neq j\).

\(\lnot C\project{B}\land\mI\) is unsat:
Again if $i$ is in \B then $i=x$ follows from \B by the choice of \(x\).
For a \B-path \(\pi\), we can conclude \(\select{s_1}{x}=\select{s_2}{x}\) by applying rule \eqref{eq:rule_readoverweakeq} and using the  index disequalities in \(\lnot C\project{B}\) and \(F^B_{\pi}(x)\).
For an \A-path \(\pi\), \(\select{s_1}{x}=\select{s_2}{x}\) (or, in cases (ii) and (iii), \(\EQ{x_{\select{a}{i}\select{b}{j}}}{\select{s}{x}}\)) follows from \(\mI_{\pi}\) using the \B-local index disequalities and \(i=x\) to show that \(F^A_\pi(x)\) cannot hold.
Transitivity and the \B-projection of \(\select{a}{i}\neq\select{b}{j}\) lead to a contradiction.
If \(i=j\) is \A-local, \(i\) is the shared term, and \(\select{b}{j}\) is in \B, the conjunct \(i=j\) in \I is needed here.
\qed
\end{proof}

%% file: subsection-interpolants-readoverweakeq-shared-a.tex
\subsection{Shared Term for $i=j$ and $\select{a}{i}=\select{b}{j}$ is \A-local}
\label{sec:interpolants_readoverweakeq_shared_a}

If there exists a shared index for \(i=j\) and \(\select{a}{i}=\select{b}{j}\) is \A-local, we build disequalities for the \B-paths instead of equalities for the \A-paths.
This corresponds to obtaining the interpolant of the inverse problem $(B,A)$ by Sect.~\ref{sec:interpolants_readoverweakeq_shared} and negating the resulting formula.
Only the \tEQ terms are not negated because of the asymmetry of the projection of mixed disequalities.

\begin{lemma}\label{the:correctness_readoverweakeq_shared_a}
  Using the definitions of $F^A_\pi$ and $F^B_\pi$ from the previous section,
  if $x$ is a shared term for $i=j$ and $\select{a}{i} = \select{b}{j}$ is \A-local,
  then a partial interpolant of the lemma $\Cond{a \weakeqi b} \land i = j \rightarrow \select{a}{i} = \select{b}{j}$
  is
\begin{equation*}\label{eq:interpolants_readoverweakeq_shared_a}
\mI\equiv\bigvee_{(\pi:s_1\weakeqi s_2)\in\text{\B-paths}}\hspace{-1.5em}(s_1[x] \neq s_2[x] \land F^B_\pi(x)) \quad\lor\bigvee_{\pi\in\text{\A-paths}} \hspace{-1em} F^A_\pi(x) \quad[{}\lor i\neq j]\enspace.
\end{equation*}
\end{lemma}

%% file: subsection-interpolants-readoverweakeq-blocal.tex
\subsection{Both $i$ and $j$ are \B-local}
\label{sec:interpolants_readoverweakeq_blocal}

When both \(i\) and \(j\) are \B-local (or both \A-local), we may not find a shared term for the index where \(a\) and \(b\) should be equal.
Instead we use the \tdiff function to express all indices where \(a\) and \(b\) differ.
For instance, if \(a=\store{\store{b}{i}{v}}{j}{w}\) for arrays \(a, b\) with $a[j] \neq b[j]$,  then \(\diff{a}{b}=j\) or \(\diff{a\rewrite[\smash{1}] b}{b}=j\) hold.

\begin{example}
\label{the:motivatingexample_readoverweakeq_blocal}
Consider the following conflict:
\begin{equation*}
a=s_1 \land \store{s_1}{k}{v} = s_2 \land s_2=b \land i\neq k \land i=j \land \select{a}{i}\neq\select{b}{j}
\end{equation*}
where \(a\), \(b\), \(i\), \(j\) are \B-local, \(k\), \(v\) are \A-local, and \(s_1\), \(s_2\) are shared.
Splitting the mixed disequality \(i\neq k\) as described in Sect.~\ref{sec:prooftreepreservinginterpolation} yields the interpolation problem
\begin{align*}
A &: \store{s_1}{k}{v} = s_2 \land \EQ{x_{ik}}{k}\\
B &: a=s_1 \land s_2=b \land \lnot\EQ{x_{ik}}{i} \land i=j \land \select{a}{i}\neq\select{b}{j} 
\enspace.
\end{align*}
An interpolant should reflect the information that \(s_1\) and \(s_2\) can differ at most at one index satisfying the \tEQ term.
Using \tdiff, we can express the interpolant
\begin{equation*}
  \mI \equiv (s_1=s_2 \lor \EQ{x_{ik}}{\diff{s_1}{s_2}})
  \land s_1\rewrite[\smash{1}] s_2=s_2\enspace.
\end{equation*}
\end{example}

To generalize this idea, we define inductively over \(m\geq 0\) for the arrays \(a\) and \(b\), and a formula \(F(\cdot)\) with one free parameter:
\begin{equation*}
\begin{aligned}
\weq{a}{b}{0}{F(\cdot)}	&\equiv a=b\\
\weq{a}{b}{m+1}{F(\cdot)} &\equiv (a=b \lor F(\diff{a}{b})) \: \land\:
\weq{a\smashrewrite b}{b}{m}{F(\cdot)}
\enspace.
\end{aligned}
\end{equation*}
The formula \(\weq{a}{b}{m}{F(\cdot)}\) states that arrays \(a\) and \(b\) differ at most at \(m\) indices and that each index $i$ where they differ satisfies the formula \(F(i)\).

\paragraph{Algorithm.}
For an \A-path \(\pi:s_1\weakeqi s_2\), we count the number of stores $|\pi|:=|\Stores{\pi}|$.
Each index $i$ where $s_1$ and $s_2$ differ must satisfy \(F^A_\pi(i)\) as defined in Sect.~\ref{sec:interpolants_readoverweakeq_shared}.
There is nothing to do for \B-paths.

\begin{lemma}\label{the:correctness_readoverweakeq_blocal}
  A partial interpolant of
  the lemma $\Cond{a \weakeqi b} \land i = j \rightarrow \select{a}{i} = \select{b}{j}$
  with \B-local $i$ and $j$ is
\begin{equation*}\label{eq:interpolants_readoverweakeq_blocal}
\mI \equiv \bigwedge_{(\pi:s_1\weakeqi s_2)\in\text{A-paths}}\hspace{-1.5em}\bigweq{s_1}{s_2}{|\pi|}{F^A_\pi(\cdot)}\enspace.
\end{equation*}
\end{lemma}

\begin{proof}
  The symbol condition holds by the same argument as in Lemma~\ref{the:correctness_readoverweakeq_shared}.

\(\lnot C\project{A}\) implies \(\mI\):
Let \(\pi:s_1\weakeqi s_2\) be an \A-path on \(P\).
The path $\pi$ shows that $s_1$ and $s_2$ can differ at most at \(|\pi|\) indices, hence \(s_1\rewrite[\smash{|\pi|}] s_2 = s_2\) follows from \(\lnot C\project{A}\).
If \(s_1\rewrite[m] s_2 \neq s_2\) holds for $m<|\pi|$, then \(\diff{s_1\rewrite[m] s_2}{s_2} = k\) for some \(k \in \Stores{\pi}\).
If \(i \neq k\) is \A-local, then $k=k$ holds trivially, if \(i \neq k\) is mixed, then \(\EQ{x_{ik}}{k}\) is part of $\lnot C \project{A}$.
Hence, \(s_1\rewrite[m] s_2 = s_2 \lor F^A_\pi(\diff{s_1\rewrite[m] s_2}{s_2})\) holds for all $m <|\pi|$.
This shows \weq{s_1}{s_2}{|\pi|}{F^A_\pi(\cdot)}.

\(\lnot C\project{B}\land\mI\) is unsat:
For every \B-path \(\pi:s_1\weakeqi s_2\) on \(P\), we get \(\select{s_1}{i}=\select{s_2}{i}\) with \eqref{eq:rule_readoverweakeq}.
For every \A-path \(\pi:s_1\weakeqi s_2\), \(\mI\) implies that \(s_1\) and \(s_2\) differ at finitely many indices which all satisfy \(F^A_{\pi}(\cdot)\). 
The disequalities and \B-projections in \B imply that \(i\) does not satisfy \(F^A_{\pi}(i)\), and therefore \(\select{s_1}{i}=\select{s_2}{i}\).
Then \(\select{a}{i}=\select{b}{i}\) holds by transitivity, in contradiction to \(\select{a}{i}\neq\select{b}{j}\) and \(i=j\) in \B.
\qed
\end{proof}

%% file: subsection-interpolants-readoverweakeq-alocal.tex
\subsection{Both $i$ and $j$ are \A-local}
\label{sec:interpolants_readoverweakeq_alocal}

The interpolant is dual to the previous case and we define the dual of weq for arrays \(a,b\), a number \(m\geq 0\) and a formula \(F\):
\begin{align*}
\nweq{a}{b}{0}{F(\cdot)}  &\equiv a\neq b\\
\nweq{a}{b}{m+1}{F(\cdot)}&\equiv (a\neq b \land F(\diff{a}{b})) \:
\lor \: \nweq{a\smashrewrite b}{b}{m}{F(\cdot)}
\enspace.
\end{align*}
The formula \(\nweq{a}{b}{m}{F(\cdot)}\) expresses that either one of the first \(m\) indices $i$ found by stepwise rewriting \(a\) to \(b\) satisfies the formula \(F(i)\), or \(a\) and \(b\) differ at more than \(m\) indices.
Like in Sect.~\ref{sec:interpolants_readoverweakeq_shared_a}, the lemma interpolant is dual to the one computed in Sect.~\ref{sec:interpolants_readoverweakeq_blocal}.

\begin{lemma}\label{the:correctness_readoverweakeq_alocal}
  A partial interpolant of
  the lemma $\Cond{a \weakeqi b} \land i = j \rightarrow \select{a}{i} = \select{b}{j}$
  with \A-local $i$ and $j$ is
\(
\mI\equiv\bigvee_{(\pi:s_1\weakeqi s_2)\in\text{B-paths}} \nweq{s_1}{s_2}{|{\pi}|}{F^B_{\pi}(\cdot)}
\).
\end{lemma}

\begin{theorem}\label{the:correctness_readoverweakeq}
For all instantiations of the rule~\eqref{eq:rule_readoverweakeq}, quantifier-free interpolants can be computed as described in Sects.~\ref{sec:interpolants_readoverweakeq_shared}--\ref{sec:interpolants_readoverweakeq_alocal}.
\end{theorem}

%% file: section-weakeqext.tex
\section{Interpolants for Weakeq-Ext Lemmas}
\label{sec:interpolants_weakeqext}

A conflict corresponding to a lemma of type \eqref{eq:rule_weakeqext} is of the form
\[ \Cond{a\weakpath b} \land
\bigwedge_{i\in\Stores{P}}\hspace{-1.5em} \Cond{a\weakcongi b} \land a \neq b
\enspace.
\]
The main path $P$ shows that $a$ and $b$ differ at most at the indices in $\Stores{P}$, and $a \weakcongi b$ (called \(i\)-path as of now) shows that $a$ and $b$ do not differ at index $i$.

To compute an interpolant, we summarize the main path by weq (or nweq) terms to capture the indices where \(a\) and \(b\) can differ, and include summaries for the \(i\)-paths that are similar to the interpolants in Sect.~\ref{sec:interpolants_readoverweakeq}.
The \(i\)-paths can contain a select edge \(a' \selectedgei[k_1,k_2] b'\) where  $\select{a'}{k_1} \strongeq \select{b'}{k_2}$, \(i\strongeq k_1\), and \(i\strongeq k_2\).
In the \B-local case~\ref{sec:interpolants_readoverweakeq_blocal}, \B-local select edges make no difference for the construction, as the weq formulas are built over \A-paths, and analogously for the \A-local case~\ref{sec:interpolants_readoverweakeq_alocal}. 
However, if there are \A-local select terms \(\select{a'}{k}\) in the \B-local case or vice versa, then $k$ is shared or the index equality
\(i=k\) is mixed and we can use $k$ or the auxiliary variable \(x_{ik}\) and proceed as in the cases where there is a shared term. 

We have to adapt the interpolation procedures in Sects.~\ref{sec:interpolants_readoverweakeq_shared} and \ref{sec:interpolants_readoverweakeq_shared_a} by adding the index equalities that pertain to a select edge, analogously to the index disequality for a store edge.
More specifically, we add to \(F^A_\pi(x)\) a disjunct \(x\neq k\) for each \B-local \(i=k\) on an \A-path, and \(x\neq x_{ik} \) for each mixed \(i=k\).
Here, \(x\) is the shared term for the \(i\)-path index \(i\).
For \B-paths we add to \(F^B_\pi(x)\) a conjunct \(x=k\) for each \A-local \(i=k\) and \(x=x_{ik}\) for each mixed \(i=k\).
Moreover, if there is a mixed select equality \(\select{a'}{k_1}=\select{b'}{k_2}\) on the \(i\)-path, the auxiliary variable \(x_{{\select{a'}{k_1}}{\select{b'}{k_2}}}\) is used in the summary for the subpath instead of $\select{s}{x}$, i.e., we get a term of the form \(\select{s_1}{x} = x_{{\select{a'}{k_1}}{\select{b'}{k_2}}}\) in \ref{sec:interpolants_readoverweakeq_shared}, and analogously for \ref{sec:interpolants_readoverweakeq_shared_a}.

\iflongversion{
\begin{lemma}
  Instead of a weak equivalence \(a \weakeqi b\), let the lemmas in Sect.~\ref{sec:interpolants_readoverweakeq} contain a weak congruence \(a \weakcongi b\).
  Then the modified methods described above give correct partial interpolants.
\end{lemma}
\begin{proof}
If the  weak congruence \(a \weakcongi b\) does not contain a select edge, the methods in Sect.~\ref{sec:interpolants_weakeqext} can be used unchanged.
We assume now that the weak congruence path \(a \weakcongi b\) consists of two weak equivalence paths \(a \weakeqi a'\) and \(b' \weakeqi b\) and the select edge \(a' \selectedgei[k_1,k_2] b'\) with \(\select{a'}{k_1} \strongeq \select{b'}{k_2}\), \(i\strongeq k_1\), and \(i\strongeq k_2\).

Let \I be the interpolant constructed as described above.
We give a proof for the modified methods based on \ref{sec:interpolants_readoverweakeq_shared} and \ref{sec:interpolants_readoverweakeq_blocal} where \(\select{a}{i} \neq \select{b}{i}\) is in \B or mixed; the other cases follow dually.

\(\lnot C \project{A}\) implies \(\mI\) (Case~\ref{sec:interpolants_readoverweakeq_blocal}):
The select edge \(a' \selectedgei[k_1,k_2] b'\) lies on a \B-path, and the index equalities \(i\strongeq k_1\) and \(i\strongeq k_2\) are \B-local (otherwise, we would get a shared term from the select index equalities and use case \ref{sec:interpolants_readoverweakeq_shared}).
Hence, the interpolant is constructed as in \ref{sec:interpolants_readoverweakeq_blocal} and \(\lnot C \project A\) implies \(\mI\) by Lemma~\ref{the:correctness_readoverweakeq_blocal}.

\(\lnot C \project{B}\land \mI\) is unsat (Case~\ref{sec:interpolants_readoverweakeq_blocal}):
For the \B-path \(s_1 \weakcongi s_2\) containing the select edge \(\select{a'}{k_1} \strongeq \select{b'}{k_2}\), we get \(\select{s_1}{i} = \select{a'}{i}\) and \(\select{b'}{i} = \select{s_2}{i}\) by \eqref{eq:rule_readoverweakeq} as in the original proof.
With congruence, \(\select{a'}{k_1} = \select{b'}{k_2}\), \(i=k_1\) and \(i=k_2\), it follows \(\select{s_1}{i} = \select{s_2}{i}\).
Proceeding as in the proof in Lemma~\ref{the:correctness_readoverweakeq_blocal} for the other paths yields the contradiction.

\(\lnot C \project{A}\) implies \(\mI\) (Case~\ref{sec:interpolants_readoverweakeq_shared}):
Let \(x\) be the shared term for \(i=j\) as in Lemma~\ref{the:correctness_readoverweakeq_shared}.
As in the original proof, \(i=x\) follows from the \A-part if \(i\) is in \A.
If the select edge lies on a \B-path \(\pi\) and the select index equality \(i=k\) is \A-local or mixed for \(k = k_1, k_2\), the conjunct \(x=k\) or \(x=x_{ik}\) is added to \(F^B_\pi(x)\).
In this case, \(i\) must be in \A and the equalities follow by replacing \(i\) with \(x\) in the corresponding equality or \A-projection.
If the select edge lies on an \A-path \(\pi: s_1 \weakcongi s_2\), we need to show the corresponding path summary \(\select{s_1}{x}=\select{s_2}{x}\lor F^A_\pi(x)\).
If \(F^A_\pi(x)\) does not hold,
we get \(\select{s_1}{x}=\select{a'}{x}\) and \(\select{b'}{x}=\select{s_2}{x}\) from \A as before.
Also \(x = k\) (for $k=k_1,k_2$) follows from the negation of $F^A_\pi(x)$ and the \A-projection of $i = k$.
Hence, \(\select{s_1}{x}=\select{a'}{x}=\select{a'}{k_1}=\select{b'}{k_2}=\select{b'}{x}=\select{s_2}{x}\).
If the select edge is mixed and w.l.o.g.\ \(\select{a'}{k_1}\) is \A-local, the corresponding \A-path is summarized as \(\select{s_1}{x} = x_{{\select{a'}{k_1}}{\select{b'}{k_2}}}\lor F^A_\pi(x)\).
Again we can derive \(\select{s_1}{x} = \select{a'}{x}\) and \(x = k_1\) if \(F^A_\pi(x)\) does not hold.
Using the \A-projection of the mixed select equality we get \(\select{s_1}{x} = x_{{\select{a'}{k_1}}{\select{b'}{k_2}}}\) as desired.
For the corresponding \B path we note that \(x = k_2\), resp.\ \(x = x_{{i}{k_2}}\) follow from \(i=x\) if \(i = k_2\) is \A-local or mixed.
Together with the original proof for all other paths, \(\lnot C \project A\) implies \(\mI\).

\(\lnot C \project B \land \mI\) is unsat (Case~\ref{sec:interpolants_readoverweakeq_shared}):
As in the proof for Lemma~\ref{the:correctness_readoverweakeq_shared}, \(i=x\) follows from the \B-part if \(i\) is in \B.
If the select edge lies on a \B-path \(\pi: s_1 \weakcongi s_2\), we get \(\select{s_1}{x}=\select{a'}{x}\) and \(\select{b'}{x}=\select{s_2}{x}\) as in the proof for Lemma~\ref{the:correctness_readoverweakeq_shared}.
The \B-projections of the index equalities \(i=k_1\) and \(i=k_2\) together with \(F^B_\pi(x)\) and the select equality \(\select{a'}{k_1}=\select{b'}{k_2}\) yield \(\select{s_1}{x}=\select{s_2}{x}\).
If the select edge lies on an \A-path \(\pi: s_1 \weakcongi s_2\), we can show that \(F^A_\pi(x)\) cannot hold by using the \B-projections of the index equalities and disequalities on this path.
Hence, \(\select{s_1}{x}=\select{s_2}{x}\) follows from the interpolant.
If the select edge is mixed and w.l.o.g.\ \(\select{a'}{k_1}\) is \A-local, the same argument yields \(\select{s_1}{x}=x_{{\select{a'}{k_1}}{\select{b'}{k_2}}}\).
For the \B-path starting at the mixed variable, we get \(\select{b'}{x}=\select{s_2}{x}\) as in the proof for Lemma~\ref{the:correctness_readoverweakeq_shared}.
Using the conjunct in \(F^B_\pi(x)\) and the \B-projection of the index equality \(i=k_2\) together with the \B-projection of \(\select{a'}{k_1}=\select{b'}{k_2}\) gives \(x_{{\select{a'}{k_1}}{\select{b'}{k_2}}}=\select{s_2}{x}\).
Proceeding as in the original proof for the other paths yields the contradiction.
\qed
\end{proof}
}\fi

\smallskip

For \eqref{eq:rule_weakeqext} lemmas, we distinguish three cases:
(i) \(a= b\) is in \B, (ii) \(a= b\) is \A-local, or (iii) \(a= b\) is mixed.

%% file: subsection-interpolants-weakeqext-blocal.tex
\subsection{$a= b$ is in \B}
\label{sec:interpolants_weakeqext_blocal}

If the literal \(a= b\) is in \B, the \A-paths both on the main store path and on the weak paths have only shared path ends.
Hence, we summarize \A-paths similarly to Sects.~\ref{sec:interpolants_readoverweakeq_shared}~and~\ref{sec:interpolants_readoverweakeq_blocal}.

\paragraph{Algorithm.}
Divide the main path $a\weakpath b$ into \A-paths and \B-paths.
For each \(i\in\Stores{P}\) on a \B-path, summarize the corresponding \(i\)-path as in Sects.~\ref{sec:interpolants_readoverweakeq_shared} or \ref{sec:interpolants_readoverweakeq_blocal}.
The resulting formula is denoted by \(\mI_i\).
For an \A-path $s_1\weakpath[\pi] s_2$ use a weq formula to state that
each index where $s_1$ and $s_2$ differ satisfies 
  \(\mI_{i}(\cdot)\) for some $i \in \Stores{\pi}$ where
\(\mI_{i}\) is computed as in \ref{sec:interpolants_readoverweakeq_shared} with the shared term $\cdot$ for $i=j$.  If $i$ is also shared we add $i=\cdot$ to the interpolant.

\begin{lemma}\label{the:correctness_weakeqext_blocal}
  The lemma
$\Cond{a\weakpath b} \land
  \bigwedge_{i\in\Stores{P}} \Cond{a\weakcongi b} \rightarrow a = b$
  where $a = b$ is in $B$ has the partial interpolant
\begin{equation*}\label{eq:interpolants_weakeqext_blocal}
  \mI\equiv\bigwedge_{\substack{i\in\Stores{\pi}\\\pi\in\text{\B-paths}}} \hspace{-1em} \mI_i\quad \land \bigwedge_{(s_1\weakpath[\pi]s_2)\in\text{\A-paths}}\hspace{-1em}\bigweq{s_1}{s_2}{|{\pi}|}{\bigvee_{i\in\Stores{\pi}} \hspace{-1em} \big(\mI_{i}(\cdot)\;[{}\land i = \cdot]\big)}\enspace.
\end{equation*}

\end{lemma}

\begin{proof}
The path summaries \(\mI_i\) fulfill the symbol conditions, and the boundary terms \(s_1,s_2\) used in the weq formulas are guaranteed to be shared.

\(\lnot C\project{A}\) implies \(\mI\):
By Sects.~\ref{sec:interpolants_readoverweakeq_shared} and \ref{sec:interpolants_readoverweakeq_blocal}, \(\Cond{a\weakcongi b} \project{A}\) implies \(\mI_{i}\) for \(i\in\Stores{\pi}\) where \(\pi\) is a \B-path on \(P\).
For an \A-path \(s_1\weakpath[\pi]s_2\) on \(P\), we know that \(s_1\) and \(s_2\) differ at most at \(|\pi|\) positions, namely at the indices \(i \in \Stores{\pi}\).
Each index satisfies the corresponding \(\mI_{i}\) by Sect.~\ref{sec:interpolants_readoverweakeq_shared}.
Hence, \(\weq{s_1}{s_2}{|{\pi}|}{\bigvee_{i\in\Stores{\pi}} I_i(\cdot) [{}\land i = \cdot]}\) holds.

\(\lnot C\project B \land\mI\) is unsat:
We first note that if \(a\) and \(b\) differ at some index \(i\), there must be an \A-path or a \B-path \(s_1\weakpath[\pi] s_2\) on the main path, such that \(s_1\) and \(s_2\) also differ at index \(i\).
We show that no such index exists.
For a \B-path \(s_1\weakpath[\pi] s_2\), \(s_1\) and \(s_2\) can only differ at \(i\in\Stores{\pi}\).
But for every \(i\in\Stores{\pi}\), we get \(\select{a}{i}=\select{b}{i}\) from \(I_i\) as in Lemma~\ref{the:correctness_readoverweakeq_shared} resp.\ \ref{the:correctness_readoverweakeq_blocal}.
For an \A-path \(s_1\weakpath[\pi] s_2\), the interpolant contains \(\weq{s_1}{s_2}{|{\pi}|}{\bigvee_{i\in\Stores{\pi}} (I_i(\cdot)[{}\land i = \cdot])}\).
Thus, if \(s_1\) and \(s_2\) differ at some index \(i'\), the interpolant implies \(I_i(i')\) for some index \(i\in\Stores{\pi}\) and additionally $i=i'$ if $i$ is shared.
Together with \(\Cond{a\weakcongi b} \project{B}\) this implies \(\select{a}{i'}=\select{b}{i'}\) as in the proof of Lemma~\ref{the:correctness_readoverweakeq_shared}.
This shows that there is no index where \(a\) and \(b\) differ, but this contradicts \(a\neq b\) in \(\lnot C\project{B}\).
\qed
\end{proof}

%% file: subsection-interpolants-weakeqext-alocal.tex
\subsection{$a= b$ is \A-local}
\label{sec:interpolants_weakeqext_alocal}

The case where \(a= b\) is \A-local is similar with the roles of $A$ and $B$ swapped.
For each \(i\in\Stores{\pi}\) on an \A-path \(\pi\) on \(P\), interpolate the corresponding \(i\)-path as in Sects.~\ref{sec:interpolants_readoverweakeq_shared_a}~or~\ref{sec:interpolants_readoverweakeq_alocal} and obtain $I_i$.
For each \(i\in\Stores{\pi}\) on a \B-path \(\pi\) on \(P\), interpolate the corresponding \(i\)-path as in Sect.~\ref{sec:interpolants_readoverweakeq_shared_a} using $\cdot$ as shared term and obtain $I_i(\cdot)$.

\begin{lemma}\label{the:correctness_weakeqext_alocal}
  The lemma
$\Cond{a\weakpath b} \land
  \bigwedge_{i\in\Stores{P}} \Cond{a\weakcongi b} \rightarrow a = b$
  where $a = b$ is \A-local has the partial interpolant
\begin{equation*}\label{eq:interpolants_weakeqext_alocal}
 \mI\equiv\bigvee_{\substack{i\in\Stores{\pi}\\ \pi\in\text{\A-paths}}} \hspace{-1em} \mI_i\quad \lor \bigvee_{(s_1 \weakpath[\pi] s_2)\in\text{\B-paths}}\hspace{-1em} \bignweq{s_1}{s_2}{|{\pi}|}{\bigwedge_{i\in\Stores{\pi}} \hspace{-1em}\big(\mI_{i}(\cdot)\;[{}\lor i \neq \cdot]\big)}\enspace.
\end{equation*}
\end{lemma}

%% file: subsection-interpolants-weakeqext-mixed.tex
\subsection{$a= b$ is mixed}
\label{sec:interpolants_weakeqext_mixed}

If \(a= b\) is mixed, where w.l.o.g.\ \(a\) is \A-local, the outer \A- and \B-paths end with \A-local or \B-local terms respectively.
The auxiliary variable \(x_{ab}\) may not be used in store or select terms, thus we first need to find a shared term representing \(a\) before we can summarize \A-paths.

\begin{example}
\label{the:motivatingexample_weakeqext_mixed}
Consider the following conflict:
\begin{align*}
 &\phantom{{}\land{}}a=\store{s}{i_1}{v_1} \land b=\store{s}{i_2}{v_2} \land a\neq b  &\quad&\text{(main path)}\\
 &{}\land \select{a}{i_1}=\select{s_1}{i_1} \land b=\store{s_1}{k_1}{w_1} \land i_1\neq k_1 &&\text{($i_1$-path)}\\
 &{}\land a=\store{s_2}{k_2}{w_2}\land i_2\neq k_2 \land \select{b}{i_2}=\select{s_2}{i_2} &&\text{($i_2$-path)}
\end{align*}
where \(a\), \(i_1\), \(v_1\), \(k_2\), \(w_2\) are \A-local, \(b\), \(i_2\), \(v_2\), \(k_1\), \(w_1\) are \B-local and \(s\), \(s_1\), \(s_2\) are shared.

\goodbreak
\noindent Our algorithm below computes the following interpolant for the conflict.
\begin{equation*}
\begin{aligned}
  \mI\equiv{}& I_0(s) \lor \bignweq{s}{s_1}{2}{I_0(\store{s}{\cdot}{\select{s_1}{\cdot}}) \land \EQ{x_{{i_1}{k_1}}}{\cdot} }\\
  &\text{where }I_0(\tilde{s}) = \EQ{x_{ab}}{\tilde{s}}\land\weq{\tilde{s}}{s_2}{1}{\EQ{x_{{i_2}{k_2}}}{\cdot}}
\end{aligned}
\end{equation*}
\end{example}

\paragraph{Algorithm.}
Identify in the main path \(P\) the first \A-path \(a\weakpath[\pi_0] s_1\) and its store indices $\Stores{\pi_0} = \{i_1,\dots i_{|\pi_0|}\}$.
To build an interpolant, we rewrite $s_1$ by storing at each index $i_m$ the value $a[i_m]$.  We use $\tilde s$ to denote the intermediate arrays.
We build a formula \(\mI_{m}(\tilde{s})\) inductively over $m \leq |\pi_0|$.
This formula is an interpolant if \(\tilde{s}\) is a shared array that differs from \(a\) only at the indices \(i_1,\dots,i_m\).

For \(m=0\), i.e., \(a=\tilde{s}\), we modify the lemma by adding the strong edge \(\tilde{s}\strongedge a\) in front of all paths and summarize it using the algorithm in Sect.~\ref{sec:interpolants_weakeqext_blocal}, but drop the weq formula for the path \(\tilde{s}\strongedge a\weakpath[\pi_0] s_1\).
This yields \(\mI_{\ref{sec:interpolants_weakeqext_blocal}}(\tilde{s})\).
We define
\begin{equation*}\label{eq:interpolants_weakeqext_mixed_nostores}
\mI_{0}(\tilde{s})\equiv\EQ{x_{ab}}{\tilde{s}}\land\mI_{\ref{sec:interpolants_weakeqext_blocal}}(\tilde{s})\enspace.\end{equation*}

For the induction step we assume that \(\tilde{s}\) only differs from \(a\) at \(i_1,\dots,i_m,i_{m+1}\).
Our goal is to find a shared index term \(x\) for \(i_{m+1}\) and a shared value \(v\) for \(\select{a}{x}\).  
We use the \(i_{m+1}\)-path to conclude that \(\store{\tilde{s}}{x}{v}\) is equal to \(a\) at \(i_{m+1}\).
Then we can include \(\mI_{m}(\store{\tilde{s}}{x}{v})\) computed using the induction hypothesis.

(i) If there is a select edge on a \B-subpath of the \(i_{m+1}\)-path or if \(i_{m+1}\) is itself shared, we immediately get a shared term \(x\) for \(i_{m+1}\).
If the last \B-path $\pi^{m+1}$ on the \(i_{m+1}\)-path starts with a mixed select equality, then the corresponding auxiliary variable is the shared value \(v\).
Otherwise, $\pi^{m+1}$ starts with a shared array \(s^{m+1}\) and \(v := \select{s^{m+1}}{x}\).
We summarize the \(i_{m+1}\)-path from \(a\) to the start of $\pi^{m+1}$ as in Sect.~\ref{sec:interpolants_readoverweakeq_shared_a} and get \(\mI_{\ref{sec:interpolants_readoverweakeq_shared_a}}(x)\).
Finally, we set
\begin{equation*}\label{eq:interpolants_weakeqext_mixed_sharedindex}
\mI_{m+1}(\tilde{s}) \equiv
\mI_{\ref{sec:interpolants_readoverweakeq_shared_a}}(x)
\lor
(\mI_{m}(\store{\tilde{s}}{x}{v}) \land
F^B_{\pi^{m+1}}(x))\enspace.
\end{equation*}

(ii) Otherwise, we split the \(i_{m+1}\)-path into \(a \weakcongi[i_{m+1}] s^{m+1}\) and \(s^{m+1} \weakpath[\pi^{m+1}] b\), where \(\pi^{m+1}\) is the last \B-subpath of the $i_{m+1}$-path.
If $s_1$ and $a$ are equal at $i_{m+1}$ then also $\tilde s$ and $a$ are equal and the interpolant is simply $I_m(\tilde s)$.  If $a$ and $s^{m+1}$ differ at $i_{m+1}$, we build an interpolant from \(a \weakcongi[i_{m+1}] s^{m+1}\) as in \ref{sec:interpolants_readoverweakeq_alocal} and obtain \(\mI_{\ref{sec:interpolants_readoverweakeq_alocal}}\). 
Otherwise, $s_1$ and $s^{m+1}$ differ at $i_{m+1}$.
We build the store path $s_1 \weakpath[P'] s^{m+1}$ by concatenating $P$ and $\pi^{m+1}$. 
Using nweq on the subpaths $s \weakpath[\pi] s'$ of $P'$ we find the shared term $x$ for $i_{m+1}$.
If $\pi$ is in $A$ we need to add the conjunct $s \rewrite[|\pi|] s' = s'$ to obtain an interpolant.
We get
\begin{equation*}\label{eq:interpolants_weakeqext_mixed_nosharedindex}
\begin{aligned}
\mI_{m+1}(\tilde{s}) \equiv{}& \mI_{m}(\tilde{s}) \: \lor \: \mI_{\ref{sec:interpolants_readoverweakeq_alocal}}\quad [\text{for }a\weakcongi[i_{m+1}] s^{m+1}] \lor{}\\
&\hspace{-2em} \bigvee_{\substack{s \weakpath[\pi] s' \text{ in \(P'\)}}} \hspace{-1em}\bignweq{s}{s'}{|\pi|}{\mI_{m}(\store{\tilde{s}}{\cdot}{\select{s^{m+1}}{\cdot}})\land F^B_{\pi^{m+1}}(\cdot)} ~ [{}\land s \rewrite[|\pi|] s' = s']\enspace. \\
\end{aligned}
\end{equation*}

\begin{lemma}\label{the:correctness_weakeqext_mixed}
  The lemma
$\Cond{a\weakpath b} \land
  \bigwedge_{i\in\Stores{P}} \Cond{a\weakcongi b} \rightarrow a = b$
  where $a= b$ is mixed has the partial interpolant
  \(\mI\equiv\mI_{|\pi_0|}(s_1)\).
\end{lemma}

\iflongversion{
\begin{proof}
\(\lnot C\project{A}\) implies \(\mI\):
We assume \(\lnot C\project{A}\) holds.  Let \(a\weakpath[\smash{\pi_0}] s_1\) be the initial \A-path of the main path of $C$ and \(\Stores{\pi_0} = \{i_1, \dots, i_{|\pi_0|}\}\).
We show by induction over $m$ that \(\mI_m(\tilde{s})\) holds for all \(\tilde{s}\) that differ from \(a\) only at indices \(i_1,\dots, i_m\) and differ from \(s_1\) only at indices where \(\tilde{s}\) doesn't differ from \(a\).
Since \(\mI = \mI_{|\pi_0|}(s_1)\) and the \A-path $\pi_0$ shows that \(s_1\) and \(a\) only differ at \(i_1,\dots, i_{|\pi_0|}\), this concludes the case.
  
For $m=0$, we have to show \(\mI_0(\tilde{s})\) for $\tilde{s} = a$.  This is \(\EQ{x_{ab}}{\tilde{s}} \land \mI_{\ref{sec:interpolants_weakeqext_blocal}}\).  The first conjunct follows from $\EQ{x_{ab}}{a}$ in \(\lnot C\project{A}\) and the second conjunct follows with Lemma~\ref{the:correctness_weakeqext_blocal}.
  
For the step to $m+1$, there are two cases (i)~and~(ii) in the algorithm.
In case~(i), we have to show \(\mI_{\ref{sec:interpolants_readoverweakeq_shared_a}}(x) \lor (\mI_m(\store{\tilde{s}}{x}{v})\land F^B_\pi(x))\).
The \A-part implies $x=i_{m+1}$ due to the way $x$ was chosen.
Also $F^B_\pi(x)$ holds because the \A-part contains all conjuncts of $F^B_\pi(i_{m+1})$.
If $\mI_{\ref{sec:interpolants_readoverweakeq_shared_a}}(x)$ doesn't hold, then as in the contradiction proof of Lemma~\ref{the:correctness_readoverweakeq_shared} (note that \A and \B swap their role) $\select{a}{x} = v$ holds.
Hence, \(\store{\tilde{s}}{x}{v}\) stores the same value as $a$ at index $i_{m+1}$.
Thus, the induction hypothesis is applicable and $\mI_m(\store{\tilde{s}}{x}{v})$ holds.

In case~(ii), if $\tilde{s}$ and $a$ do not differ at $i_{m+1}$, $\mI_m(\tilde{s})$ holds by induction hypothesis.
Otherwise $\tilde{s}$ and $s_1$ store the same value at index $i_{m+1}$ different from $a$.
If $\mI_{\ref{sec:interpolants_readoverweakeq_alocal}}$ for the weak path $a\weakcongi[i_{m+1}] s^{m+1}$ does not hold, then similar as in the contradiction case of Lemma~\ref{the:correctness_readoverweakeq_blocal} (again $A$ and $B$ swap their roles), $\select{a}{i_{m+1}}=\select{s^{m+1}}{i_{m+1}}$ holds.
Thus $s_1$ and $s^{m+1}$ differ at $i_{m+1}$ and the latter array stores the same value as $a$.
Since path $P'$ runs from $s_1$ to $s^{m+1}$, on some of its subpaths $s\weakpath[\pi]s'$ the arrays $s$ and $s'$ differ at $i_{m+1}$.
If the subpath $\pi$ is \A-local, $s\rewrite[|\pi|]s'=s'$ follows from $\Cond{\pi}\project{A}$.
Finally, for the index $i_{m+1}$ the formula $\mI(\store{\tilde{s}}{\cdot}{\select{s^{m+1}}{\cdot}})$ holds by induction hypothesis and all conjuncts of $F^B_{\pi^{m+1}}(i_{m+1})$ are in $\Cond{\pi^{m+1}}\project{A}$.
Hence, \(\mI_{m+1}(\tilde{s})\) holds.
  
\(\lnot C\project{B}\land\mI\) is unsat:
We assume \(\lnot C\project{B}\) holds.  We show by induction over $m$ that \(\mI_m(\tilde{s})\) implies that \(\tilde{s}\) differs from \(s_1\) at some index \(i\) where it also differs from $b$.
Thus, \(\mI = \mI_{|\pi_0|}(s_1)\) leads to a contradiction.

For $m=0$, assume \(\EQ{x_{ab}}{\tilde{s}} \land \mI_{\ref{sec:interpolants_weakeqext_blocal}}\) holds.
The first part shows that $\tilde{s}$ and $b$ differ at some index $i$. Using Lemma~\ref{the:correctness_readoverweakeq_blocal}, we conclude that $\tilde{s}$ also differs from $s_1$ at index $i$, since this was the only part omitted when computing \(\mI_{\ref{sec:interpolants_weakeqext_blocal}}\).
This shows the induction hypothesis for $m=0$.

For the step to $m+1$, assume $\I_{m+1}(\tilde{s})$ holds.
In case (i), $\mI_{\ref{sec:interpolants_readoverweakeq_shared_a}}(x)$ cannot hold, because it leads to a contradiction with $\Cond{a\weakcongi[i_{m+1}]s^{m+1}}\project{B}$.
Hence, $\mI_m(\store{\tilde{s}}{x}{v}) \land F^B_{\pi^{m+1}}(x)$ holds.
The formulas \(F^B_{\pi^{m+1}}(x)\) and \(\Cond{\pi^{m+1}}\project{B}\) together with the choice of $v$ ensure that $\select{b}{x} = v$.
By the induction hypothesis, $\store{\tilde{s}}{x}{v}$ differs from $s_1$ at some index \(i\) where it also differs from \(b\) and this cannot be at $x$ because of $\select{b}{x} = v$.
Thus $\tilde{s}$ and $s_1$ also differ at this index.
In case (ii), if $\mI_m(\tilde{s})$ holds, we can use the induction hypothesis directly.
The disjunct $\mI_{\ref{sec:interpolants_readoverweakeq_alocal}}$ contradicts \(\Cond{a\weakcongi[i_{m+1}]s^{m+1}}\project{B}\).
If for some path $s\weakpath[\pi]s'$ the nweq disjunct holds, then we first note that \(s \rewrite[|\pi|] s' = s'\) holds, either because it is part of the disjunct or because $s \weakpath[\pi] s'$ is a \B-path of length $|\pi|$.
Hence, the nweq ensures that for some index $x$ the formula $\mI_{m}(\store{\tilde{s}}{x}{\select{s^{m+1}}{x}})\land F^B_{\pi^{m+1}}(x)$ holds.
Similar to case~(i), we get $\select{b}{x} = \select{s^{m+1}}{x}$ from $F^B_{\pi^{m+1}}(x)$ and, by induction hypothesis, $\tilde{s}$ must differ from $s_1$ at some index where it differs from $b$.
\qed
\end{proof}
}\else{
  \noindent A proof by induction over the length of the path $\pi_0$ can be found in~\cite{2018arXiv180407173H}.
}\fi

\begin{theorem}\label{the:correctness_weakeqext}
  Sects.~\ref{sec:interpolants_weakeqext_blocal}--\ref{sec:interpolants_weakeqext_mixed} give interpolants for all cases of the rule \(\eqref{eq:rule_weakeqext}\).
\end{theorem}

%% file: complexity.tex
\section{Complexity}\label{sec:complexity}

Expanding the definition of an array rewrite term $a \rewrite[k] b$ na\"ively already yields a term exponential in $k$.
This is avoided by using let expressions for common subterms.
With this optimization the interpolants for read-over-weakeq lemmas are quadratic in the worst case.
The interpolants of Sects.~\ref{sec:interpolants_readoverweakeq_shared}~and~\ref{sec:interpolants_readoverweakeq_shared_a} contain at most one literal for every literal in the lemma, so the interpolant is linear in the size of the lemma.
The interpolants of Sects.~\ref{sec:interpolants_readoverweakeq_blocal}~and~\ref{sec:interpolants_readoverweakeq_alocal} are quadratic, since expanding the definition of weq will copy the formula $F_\pi^A(\cdot)$ resp.\ $F_\pi^B(\cdot)$, for each local store edge and instantiate it with a different shared term.

\begin{example}
\label{ex:complexity_roweq}
The following interpolation problem has only quadratic interpolants.
\[
  \begin{aligned}
    A :{}& b = \store{\store{a}{i_1}{v_1}\cdots}{i_n}{v_n} \land
         p_1(i_1) \land \dots \land p_n(i_n)\\
    B :{}& a[j] \neq b[j] \land \lnot p_1(j) \land \dots \lnot p_n(j)\\
    I \equiv{} &\smtlet a_0 = a  
     \smtlet d_1 = \diff{a_0}{b} \smtlet a_1 = \store{a_0}{d_1}{\select{b}{d_1}}\\
     &\dots
     \smtlet d_n = \diff{a_{n-1}}{b} \smtlet a_n = \store{a_{n-1}}{d_n}{\select{b}{d_n}} \\
     & (p_1(d_1) \lor \dots \lor p_n(d_1) \lor a_0 = b) \land {\cdots} \\
     & (p_1(d_n) \lor \dots \lor p_n(d_n) \lor a_{n-1}=b) \land  a_n = b
  \end{aligned}\]
  There is no interpolant that is not quadratic in $n$.
  The interpolant has to imply that $p_k(i_k)$ is true for every $k$.
  There are no shared index-valued terms in the lemma.
  Hence, the only way to express the $i_k$ values using shared terms is by applying the diff operator on $a$ and $b$ and constructing diff chains as in the interpolant $I$.
  The diff operator returns one of the $i_1,\dots,i_n$ in every step, but it is not determined which one.
  Consequently, every combination $p_k(d_l)$ is needed.
\end{example}

The algorithms in Sects.~\ref{sec:interpolants_weakeqext_blocal}~and~\ref{sec:interpolants_weakeqext_alocal} produce a worst-case quadratic interpolant as they nest the linear interpolants of \ref{sec:interpolants_readoverweakeq_shared}~and~\ref{sec:interpolants_readoverweakeq_shared_a} in a weq resp.\ nweq formula, which expands this term a linear number of times.
However, the algorithm in \ref{sec:interpolants_weakeqext_mixed} is worst-case exponential in the size of the extensionality lemma. 

The following example explains why this bound is strict.
This example also shows that the method of Totla and Wies~\cite{DBLP:journals/jar/TotlaW16} is not complete.
In particular, for $n=1$ their preprocessing algorithm produces a satisfiable formula from the original interpolation problem.

\begin{example}
The following interpolation problem of size $O(n^2)$ has only interpolants of exponential size in $n$.
  \begin{align*}
    A :{} &a = \store{\store{s}{i^A_{1}}{v^A_1}\cdots}{i^A_{n}}{v^A_n} \land p(a) \land {}\\
    &\bigwedge_{j=1}^n p_j(i^A_j) \:\land\:
    \bigwedge_{j=1}^n \select{a}{i^A_{j}} = \select{s_{j}}{i^A_{j}} \land{} \\
    &\bigwedge_{j=1}^n \bigwedge_{l=0, l\neq j}^n q_j(i^A_{l}) \:\land \:
    \bigwedge_{j=1}^n t_j = \store{\store{a}{i^A_{0}}{w^A_{j0}} \dots \xcancel{\store{}{i^A_{j}}{w^A_{jj}}} \dots}{i^A_{n}}{w^A_{jn}}
  \end{align*}
  \begin{align*}
    B :{} &b = \store{\store{s}{i^B_{1}}{v^B_1}\cdots}{i^B_{n}}{v^B_n} \land \lnot p(b) \land {}\\
    &\bigwedge_{j=1}^n \bigwedge_{l=0, l \neq j}^n \lnot p_j(i^B_{l}) \:\land \:
     \bigwedge_{j=1}^n s_j = \store{\store{b}{i^B_{0}}{w^B_{j0}}\dots \xcancel{\store{}{i^B_{j}}{w^B_{jj}}}\dots}{i^B_{n}}{w^B_{jn}}  \land{}\\
    &\bigwedge_{j=1}^n \lnot q_j(i^B_{j}) \:\land\:
     \bigwedge_{j=1}^n \select{b}{i^B_{j}} = \select{t_{j}}{i^B_{j}}
  \end{align*}
  The first line of \A and the first line of \B ensure that there is a store-chain from $a$ over $s$ to $b$ of length $2n$ and $p(a)$ and $\lnot p(b)$ are used to derive the contradiction from the extensionality axiom.
  To prove that $a$ and $b$ are equal, the formulas show that they are equal at the indices $i^A_j$, $j=1,\dots,n$ (second line of $A$ and $B$).
  Here $p_j$ is used to ensure that $i^A_j$ is distinct from all $i^B_{l}$, $l\neq j$.
  Analogously the last line of $A$ and $B$ shows that $a$ and $b$ are equal at the indices $i^B_j$, $j=1,\dots,n$.
  
  Since $p(a) \land \lnot p(b)$ is essential to prove unsatisfiability, the interpolant needs to contain the term $p(\cdot)$ for some shared array term that is equal to $a$ and $b$.  
  This can only be expressed by store terms of size $n$, e.\,g., $p(\store{\store{s}{i_1}{\cdot}\cdots}{i_n}{\cdot})$ (alternatively some store term starting on $s_j$ or $t_j$ can be used).
  As in the previous example, the store indices $i_j$ can only be expressed using diff chains between shared arrays.
  For each index there is only one shared array that is guaranteed to contain the right value. 
  The diff function returns the indices in arbitrary order.
  Therefore, the interpolant needs a case for every combination of diff term and value, as it is done by the interpolant computed in Section~\ref{sec:interpolants_weakeqext_mixed}.
  This means the interpolant contains exponentially many $p(\cdot)$ terms.
\end{example}

%% file: evaluation.tex
\section{Evaluation}
We implemented the presented algorithms into
\smtinterpol~\cite{DBLP:conf/spin/ChristHN12}, an SMT solver computing sequence
and tree interpolants.
Our implementation verifies at run-time that the returned interpolants are
correct.
To evaluate the interpolation algorithm we used the \automizer software
model-checker~\cite{DBLP:conf/tacas/HeizmannCDGNMSS17} on the memory safety
track of the \svcomp~2018~\footnote{\url{https://sv-comp.sosy-lab.org/2018/}}
benchmarks.
This track was chosen because \ultimate uses arrays to model memory access.
We ran our experiments using the open-source benchmarking software
\texttt{benchexec}
~\cite{DBLP:conf/spin/0001LW15a}
on a
machine with a 3.4~GHz Intel i7-4770 CPU and set a 900~s time and a 6~GB memory
limit.
As comparison, we ran \ultimate with \zzz~\footnote{\url{https://github.com/Z3Prover/z3} in version 4.6.0 (2abc759d0)} and
SMTInterpol without array interpolation using \ultimate's built-in theory-independent interpolation scheme based on unsatisfiable cores and predicate transformers~\cite{DBLP:conf/tacas/HeizmannDLMP15}.

Table~\ref{tab:stats} shows the result. 
From the 326 benchmarks we removed 50 benchmarks which \ultimate could not parse.
The unknown results come from non-linear arithmetic (\smtinterpol), quantifiers (due to incomplete elimination in the setting \smtinterpol-NoArrayInterpol), or incomplete interpolation engine (\zzz).
Our new algorithm solves 12.6~\% more problems, and both helps to verify safety and guide the counterexample generation for unsafe benchmarks.

\begin{table}[t]
\centering
\caption{Evaluation of \automizer on the \svcomp benchmarks for memsafety running with our new interpolation engine, without array interpolation, and \zzz.
\label{tab:stats}}
\input{tab/eval}
\end{table}

%% file: tab/eval.tex
{
\begin{tabular}{@{}l>{\raggedleft}m{1.2cm}>{\raggedleft}m{1.2cm}>{\raggedleft}m{1.2cm}>{\raggedleft}m{1.2cm}r@{}}
\toprule
Setting                     & Tasks & Safe & Unsafe & Timeout & Unknown \\ 
\midrule
\smtinterpol-ArrayInterpol   & 276   & 101   & 96   & 66      & 13      \\
\smtinterpol-NoArrayInterpol & 276   &  92   & 83   & 75      & 26      \\
\zzz                          & 276   &  32   & 44   & 13      & 187     \\
\bottomrule
\end{tabular}
}

%% file: conclusion.tex
\section{Conclusion}
\label{sec:conclusion}

We presented an interpolation algorithm for the quantifier-free fragment of the theory of arrays.
Due to the technique of proof tree preserving interpolation, our algorithm also works for the combination with other theories.
Our algorithm operates on lemmas produced by an efficient array solver based on weak equivalence on arrays.
The interpolants are built by simply iterating over the weak equivalence and weak congruence paths found by the solver.
We showed that the complexity bound on the size of the produced interpolants is optimal.

In contrast to most existing interpolation algorithms for arrays, the solver does not depend on the partitioning of the interpolation problem.
Thus, our technique allows for efficient interpolation especially when several interpolants for different partitionings of the same unsatisfiable formula need to be computed.
Although it remains to prove formally that the algorithm produces tree interpolants,
during the evaluation all returned tree interpolants were correct.

\paragraph{Acknowledgement.}
We would like to thank Daniel Dietsch for running the experiments.